\documentclass{article}

\title{Regular networks are determined by their trees}

\author{Stephen J. Willson\\
		Department of Mathematics\\
		Iowa State University\\
		Ames, IA 50011 USA\\
		swillson@iastate.edu\\
}

\usepackage{amsmath}
\usepackage{amsthm}
\usepackage{amssymb}

\usepackage[dvips]{graphicx}
\usepackage{verbatim} 
\newtheorem{lem}{Lemma}[section]
\newtheorem{thm}[lem]{Theorem}
\newtheorem{cor}[lem]{Corollary}

\begin{document}

\maketitle

\vspace{10 pt}

{\bf Abstract.}  A rooted acyclic digraph $N$ with labelled leaves displays a tree $T$ when there exists a way to select a unique parent of each hybrid vertex resulting in the tree $T$.  Let $Tr(N)$ denote the set of all trees displayed by the network $N$.  In general, there may be many other networks $M$ such that $Tr(M) = Tr(N)$.   A network is regular if it is isomorphic with its cover digraph.  This paper shows that if $N$ is regular, there is a procedure to reconstruct $N$ given $Tr(N)$.  Hence if $N$ and $M$ are regular networks and $Tr(N) = Tr(M)$, it follows that $N = M$, proving that a regular network is uniquely determined by its displayed trees.

\section{ Introduction}  % Section 1

It has become common, for a given collection $X$ of taxa and given a particular gene $g$, to use phylogenetic methods to determine a phylogenetic tree $T^g$.  The extant taxa correspond to leaves of the trees, while internal vertices correspond to ancestral species.  The arcs correspond to genetic change, typically by mutations in the DNA such as substitutions, insertions, and deletions.  Common methods for determining the trees include maximum likelihood, maximum parsimony, and neighbor-joining, but many other methods are also utilized.

Commonly the use of a different gene $h$ for the same collection $X$ of taxa results in a tree $T^{h}$ that differs from $T^g$.  Indeed, many different trees arise for different genes $g$ but the same $X$.  For example \cite{rok03} utilized 106 orthologs common to seven species of yeast and an outgroup.  The collection of 106 maximum-parsimony trees and 106 maximum-likelihood trees included more than 20 different robustly supported topologies.   While \cite{rok03} concatenated the data to try to achieve resolution, \cite{hol04} employed consensus networks to display the incompatibilities that existed among the trees.  

One hypothesis to explain the deviations of such gene trees from a single ``species tree" is to assume ``lineage sorting".  In this model a single species tree is seen as a kind of pipeline containing populations with significant genetic diversity; the genes actually fixate at locations that need not coincide with the speciation events in the species tree.  Hence the genes do not necessarily follow the species tree.  Coalescence methods such as \cite{ros02}, \cite{deg06}, \cite{ros07}  utilize this approach.  For example, \cite{deg06} shows that the most likely gene tree need not coincide with the species tree.  Much of the resulting diversity, however, makes use of short branch-lengths separating some speciation events in the species tree.  

Another hypothesis to explain the deviations of such gene trees from a single ``species tree" is to assume that evolution actually occurs on networks that are not necessarily trees.  Besides mutation events, these networks could include such additional reticulation events as hybridization or lateral gene transfer.   General frameworks are discussed in \cite{ban92}, \cite{bar04}, \cite{mor04}, and \cite{nak05}.

Even if the underlying species relationships are given by a network, the evolution of an individual gene might best be described by a tree.  The idea is that, at a hybridization event, some genes would be inherited from one parent species, and other genes from another parent species.  Suppose, for example, the underlying species network is $M$ in Figure 1.   Species 2 is hybrid with parental species $B$ and $C$.  If a particular gene in 2 is inherited from $B$, then the correct description of the inheritance of that gene would be tree $b$ in Figure 2.  If instead a gene in 2 is inherited from $C$, then the correct description for that gene would be tree $c$ in Figure 2.  Thus we would expect to see both trees $b$ and $c$ among the various gene trees.  Trees $b$ and $c$ are said to be \emph{displayed} by the network.  On the other hand, tree $d$ in Figure 2 is not displayed by $M$, so we would not expect a gene to evolve according to $d$ under these assumptions.

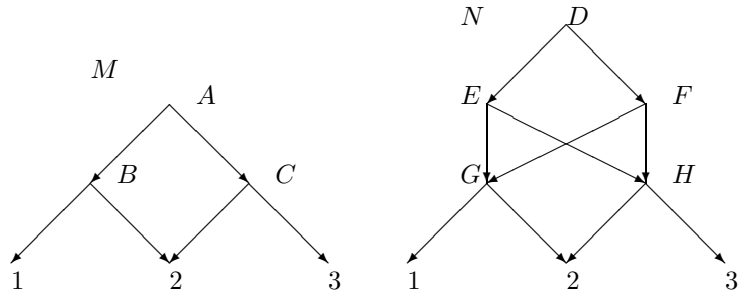
\begin{figure}[!htb]  %1
\begin{center}

\begin{picture}(290,130) 

\put(70,70) {\vector(-1,-1){30}}
\put(40,40) {\vector(-1,-1){30}}
\put(70,70) {\vector(1,-1){30}}
\put(100,40) {\vector(1,-1){30}}
\put(100,40) {\vector(-1,-1){30}}
\put(40,40) {\vector(1,-1){30}}
\put(10,0){1}
\put(70,0){2}
\put(130,0){3}
\put(50,40){$B$}
\put(110,40){$C$}
\put(80,70){$A$}
\put(40,80){$M$}

\put(220,100) {\vector(1,-1){30}}
\put(220,100) {\vector(-1,-1){30}}
\put(190,70) {\vector(0,-1){30}}
\put(190,70) {\vector(2,-1){60}}
\put(250,70) {\vector(-2,-1){60}}
\put(250,70) {\vector(0,-1){30}}
\put(190,40) {\vector(-1,-1){30}}
\put(190,40) {\vector(1,-1){30}}
\put(250,40) {\vector(-1,-1){30}}
\put(250,40) {\vector(1,-1){30}}
\put(160,0){1}
\put(220,0){2}
\put(280,0){3}
\put(180,40){$G$}
\put(260,40){$H$}
\put(180,70){$E$}
\put(260,70){$F$}
\put(220,100){$D$}
\put(180,100){$N$}

\end{picture}

\caption{  Two phylogenetic networks with base-set $X = \{1,2,3\}$.}
\end{center}
\end{figure}

\begin{figure}[!htb]  %2
\begin{center}

\begin{picture}(320,160) 

\put(70,150) {\vector(-1,-1){30}}
\put(70,150) {\vector(1,-1){30}}
\put(40,120) {\vector(-1,-1){30}}
\put(40,120) {\vector(1,-1){30}}
\put(100,120) {\vector(0,-1){30}}
\put(80,150){$A$}
\put(50,120){$B$}
\put(110,120){$C$}
\put(10,80){1}
\put(70,80){2}
\put(100,80){3}
\put(40,150){$a$}

\put(190,150) {\vector(-1,-1){30}}
\put(190,150) {\vector(1,-2){30}}
\put(160,120) {\vector(-1,-1){30}}
\put(160,120) {\vector(1,-1){30}}
\put(200,150){$A$}
\put(150,120){$B$}
\put(130,80){1}
\put(190,80){2}
\put(220,80){3}
\put(160,150){$b$}

\put(40,70) {\vector(-1,-2){30}}
\put(40,70) {\vector(1,-1){30}}
\put(70,40) {\vector(-1,-1){30}}
\put(70,40) {\vector(1,-1){30}}
\put(50,70){$A$}
\put(80,40){$C$}
\put(10,0){1}
\put(40,0){2}
\put(100,0){3}
\put(10,60){$c$}

\put(160,70) {\vector(-1,-2){30}}
\put(160,70) {\vector(1,-1){30}}
\put(190,40) {\vector(-1,-1){30}}
\put(190,40) {\vector(1,-1){30}}
\put(130,0){2}
\put(160,0){1}
\put(220,0){3}
\put(130,60){$d$}

\put(280,100) {\vector(-1,-1){30}}
\put(280,100) {\vector(1,-1){30}}
\put(250,70) {\vector(0,-1){30}}
\put(250,40) {\vector(0,-1){30}}
\put(250,70) {\vector(1,-1){30}}
\put(280,40) {\vector(0,-1){30}}
\put(280,40) {\vector(1,-1){30}}
\put(250,0){1}
\put(280,0){2}
\put(310,0){3}
\put(240,40){$G$}
\put(290,40){$H$}
\put(240,70){$E$}
\put(320,70){$F$}
\put(290,100){$D$}
\put(250,100){$e$}

\end{picture}

\caption{ Some trees related to Figure 1.   Both $M$ and $N$ display trees $b$ and $c$ but not $d$.  }
\end{center}
\end{figure}
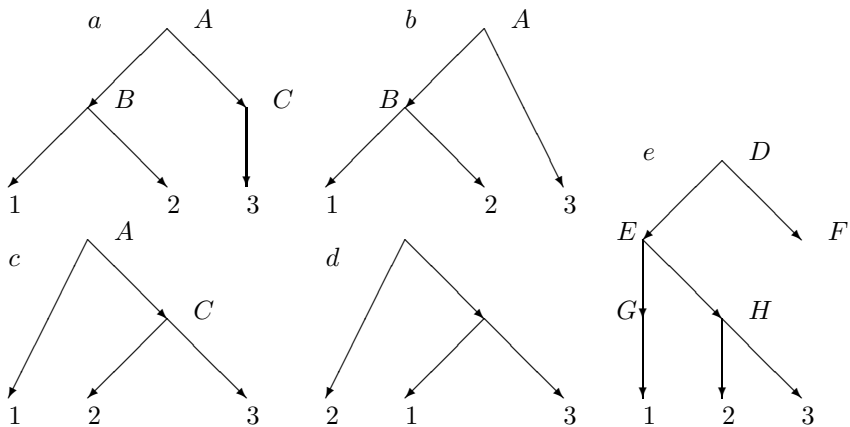

The assumption that the underlying description of evolutionary history is a network rather than a tree raises the fundamental problem of reconstructing a network from data.  Suppose that a collection of gene trees for the same set X of taxa is known.  Can the underlying network be uniquely reconstructed?

If $M$ is a network, let $Tr(M)$ denote the set of rooted trees displayed by $M$.  Figure 1 shows two distinct networks $M$ and $N$ such that $Tr(M) = Tr(N) = \{b,c\}$.  This example represents a common situation.  In general there may be many networks that display exactly the same trees.  

One approach has been to seek a network that displays a collection of trees and which has the fewest hybridization events.  This problem was proved to be NP-hard \cite{bor07}.  Various special cases with additional hypotheses on the networks have also been studied, such as  \cite{wan01}, \cite{gel04a}, \cite{hus05}, \cite{hus07}.   

A different approach has been to make assumptions on the properties of an allowable phylogenetic network.  It would be desirable to have a class of phylogenetic networks which is biologically plausible and such that there is often a uniquely determined network of this type with certain observable properties.  It is commonly assumed that the networks are rooted acyclic digraphs \cite{str00}, \cite{mor04}, \cite{nak05}.  Restrictions that appear tractable and yield interesting results include time consistency \cite{mor04}, \cite{car08}, roughly that the parents of a hybrid be contemporaneous.  Others include restrictions on the children of vertices, for example \emph{tree-child networks} \cite{crv07} or \emph{tree-sibling networks} \cite{car08}. Certain unique reconstructions for ``normal" networks are given in \cite{wil08}.  

Baroni and Steel \cite{bs06} defined the notion of a \emph{regular} network.  The precise definition is given in Section 2.   The basic idea is as follows:  The \emph{cluster} $cl(v)$ of a vertex $v$ is the set of leaves which are descendents of $v$.  In a regular network, no two distinct vertices have the same cluster.  Moreover,  $cl(u) \subset cl(v)$ iff there is a directed path from $v$ to $u$.  In Figure 1, $M$ is regular but $N$ is not since in $N$ $cl(D) = cl(E) = cl(F) = \{1,2,3\}$.  

The main result of this paper is Theorem 3.1.  This theorem gives a method which, given $Tr(M)$ for a regular network $M$ uniquely reconstructs the network $M$.  Corollary 3.2 asserts the consequence that if $M$ and $N$ are regular networks with the same leaves and $Tr(M) = Tr(N)$, then $M = N$.  Thus the entire collection of trees displayed by a regular network uniquely determines the network.  

Figure 1  shows that without the assumption of regularity, the network is not uniquely determined by the set of its displayed trees.  Another example is given in Section 5.

The proof of Theorem 3.1 is constructive.  A procedure MaximumProperChild is applied to the input $\mathcal{D} = Tr(N)$.  When $N$ is regular, the procedure outputs the network $N$ up to isomorphism.  (In fact it outputs the cover digraph of N, see \cite{bs06}.)  An example is worked in section 4, illustrating the procedure.  Also in section 4, we observe that the input need not be $\mathcal{D} = Tr(N)$ but instead might be an appropriate subset of $Tr(N)$.  Characterizing this subset, however, remains an open problem.  

It is not likely in a real biological problem that all the trees displayed by a network are known.  The number of such trees could easily grow exponentially with the number of leaves.  Hence the main theorem is primarily of theoretical interest:  For any two distinct regular phylogenetic networks there must exist a tree displayed by one but not the other.  

This situation contrasts with that in which the \emph{generalized clusters} or \emph{tree clusters} of a network are given instead of all the displayed trees.  For a network $N$, a generalized or tree cluster is any cluster of any tree $T$ displayed by $N$.  The set of all tree clusters of $N$ is denoted $TrCl(N)$.  The paper \cite{wil08} presents examples of distinct regular networks (indeed normal networks) $M$ and $N$  with the same leaf sets and which have precisely the same tree clusters; thus $TrCl(M) = TrCl(N)$ but $M$ and $N$ are not isomorphic.  The author therefore finds it somewhat surprising that, as shown in the current paper, the trees themselves do determine the network uniquely for a broad class of networks.

\section {Basics} %.  Trees displayed by a network  %2

A \emph{directed graph} or \emph{digraph} $N=(V,A)$ consists of a finite set $V = V(N)$ of \emph{vertices} and a finite set $A = A(N)$ of \emph{arcs}, each consisting of an ordered pair $(u,v)$ where $u \in V$, $v \in V$, $u\neq v$, interpreted as an arrow from $u$ (the \emph{parent}) to $v$ (the \emph{child}).  There are no multiple arcs and no loops.   A \emph{directed path} is a sequence $u_0, u_1, \cdots, u_k$ of vertices such that for $i = 1, \cdots, k$, $(u_{i-1}, u_i) \in A$.  The \emph{length} of the path is $k$ and the path is \emph{trivial} if $k = 0$.  The graph is \emph{acyclic} if there is no nontrivial directed path starting and ending at the same point.  Write $u \leq_N  v$  or more informally $u \leq v$  in $N$ if there is a directed path starting at $u$ and ending at $v$.  Write $u < v$ if $u \leq v$ and $u \neq v$.   If the graph is acyclic, it is easy to see that $\leq$  is a partial order on $V$.  

A vertex $r$ is a \emph{root} of the directed acyclic graph $(V,A)$  if, for all $v \in V$, $r\leq v$.  The network is \emph{rooted} if it has a root.  Clearly there can be at most one root.  

The \emph{indegree} of vertex $u$ is the number of $v \in V$ such that $(v,u) \in A$.  The \emph{outdegree} of $u$ is the number of $v \in V$ such that $(u,v) \in A$.  If $N$ is rooted at $r$ then $r$ is the only vertex of indegree 0.  A \emph{leaf} is a vertex of outdegree 0. A \emph{normal} (or \emph{tree}) vertex is a vertex of indegree at most 1.  A \emph{hybrid} vertex (or \emph{recombination vertex} or \emph{reticulation node}) is a vertex of indegree at least 2.  

Let $X$ be a set.  The cardinality of $X$ will be denoted $|X|$.  In biological terms we consider the members of $X$ to be a specific collection of biological species.  We call $X$ the \emph{base-set} of the directed graph $N = (V,A)$ if there is a given one-to-one relationship between $X$ and the subset $L(N) \subseteq V$ consisting of the leaves of $N$.  Thus we identify the leaves of $N$ with the members of $X$.  The interpretation of $X$ is that its members correspond to taxa on which direct measurements may be made, while $N$ describes a proposed evolutionary history giving rise to these taxa.  The leaves correspond to extant taxa so direct measurements are possible.  Typically one taxon is included which is an outgroup----an extant species clearly on a separate evolutionary track from all other taxa.  Hence the root is located as the attachment vertex of the outgroup taxon.  

In this paper a \emph{(phylogenetic) network} $N = (V,A,r,X)$ is an acyclic digraph $(V,A)$ with root $r$ and base-set $X$.  Two networks $N = (V,A,r,X)$ and $M = (V',A',r',X)$ are \emph{isomorphic}, $N \cong M$, iff there is a bijection $\phi: V \to V'$ such that for all $x \in X$, $\phi(x) = x$, and $(u,v) \in A$ iff $(\phi(u),\phi(v)) \in A'$.  

Let $N= (V,A,r,X)$ be a phylogenetic network.  Let $\mathcal{P}(X)$ denote the set of all subsets of $X$.  For $v \in V$, define the (full) \emph{cluster} of $v$ in $N$ by
$cl(v,N) = \{x \in X: v\leq x\}$.  
It is clear that for each $v \in V$, $cl(v,N) \in \mathcal{P}(X)$.  Define for each phylogenetic network $N$ with base-set $X$, 
	$cl_N : V \to \mathcal{P}(X)$
by $cl_N(v) = cl(v,N)$. 

The following properties of the clusters are basic:\\
(1) For $v \in V$, $cl_N(v)$ is nonempty.\\
(2) If $u \leq_N v$, then $cl_N(v) \subseteq cl_N(u)$.\\
(3) $cl_N(r) = X$.\\
(4) If $x \in X$, then $cl_N(x) = \{x\}$.  

Note that (1) follows since a maximal path must end at a leaf and every leaf lies in $X$.  Morever (2) follows since $\leq_N$ is a partial order.  In particular, if $(u,v)$ is an arc of $N$, then $cl_N(v) \subseteq cl_N(u)$.  Also, (3) follows since for each $x \in X$, we have $r \leq x$, and (4) follows since each $x \in X$ satisfies that $x$ is a leaf.  

The clusters $X$ and $\{x\}$ for $x \in X$ are called the \emph{trivial clusters} since they occur in each network.  Any other clusters will be called \emph{nontrivial}.  

Given the network $N=(V,A,r,X)$, we may let $\mathcal{C}(N) = \{cl_N(v): v \in V\} \subseteq \mathcal{P}(X)$.  The \emph{cover digraph} of $N$ is the digraph $(W,E)$ where \\
(1) $W = \mathcal{C}(N)$, and \\
(2) there is an arc $(B,C) \in E$  for $B$ and $C$ in $W$ iff\\
\indent(2a) $C \subset B$; and \\
\indent(2b) there is no $D \in W$ such that $C \subset D \subset B$.\\  
Note  the root is $X = cl_N(r)$ because for all $x \in X$, $r \leq x$.
Note since the members of $X$  are the leaves of $N$, it follows for each $x \in X$, $cl_N(x) = \{x\}$ so the leaves of the cover digraph are the singleton sets $\{x\}$ for $x \in X$.  Hence the leaves may be identified with the members of $X$ and the root $r$ with $X$.  

Baroni and Steel \cite{bs06} defined a \emph{regular} network to be a network which is isomorphic with its cover digraph.  Following is an equivalent description:  The phylogenetic network $N = (V,A,r,X)$ is \emph{regular} provided\\
(1) $cl_N: V \to \mathcal{P}(X)$  is one-to-one; and \\
(2) there is an arc $(u,v) \in A$ iff $cl_N(v) \subset cl_N(u)$ and there is no $w \in V$ such that
$cl_N(v) \subset cl_N(w) \subset cl_N(u)$.

Let $N = (V,A,r,X)$ be a phylogenetic network.  A \emph{parent map} for $N$ is a map
$p: V-\{r\} \to V$
such that for each $v \in V$, $v\neq r$, $p(v)$ is a parent of $v$, i.e., $(p(v),v) \in A$.
Since the root $r$ is unique, it is clear that if $v \in V$, $v\neq r$, then $v$ has a parent.  Note that if $v$ is normal and $v\neq r$, then $v$ has exactly one parent $q$, so all parent maps $p$ will satisfy $p(v) = q$.  When $v$ is hybrid, however, there are at least two parents of $v$.

Let $Par(N)$ denote the collection of parent maps for the network $N$.  Let $i(v,N)$ denote the indegree of $v$ in the network $N$.   Then the number of distinct parent maps is clearly
$|Par(N)| = \prod [i(v,N): v \in V, v\neq r]$.

For any parent map $p$ for $N =(V,A,r,X)$ construct a new network $N_p = (V,E,r,X)$ as follows:  The vertex set, root, and base-set are the same as for $N$.  The arc set $E$ consists of all arcs of the form $(p(v),v)$ where $v \in V$, $v\neq r$.  Thus $E \subseteq A$.

Each vertex $v$ other than $r$ has exactly one parent in $N_p$; i.e., $i(v,N_p) = 1$.  Hence $N_p$ is a rooted tree.  It is quite possible that $v$ has outdegree 1 as well as indegree 1, but such vertices are often suppressed in a rooted tree.  We will therefore consider two kinds of simplification to change  $N_p$ into a rooted tree in standard form.  

Type 1:  Suppress a vertex with outdegree 1.  More specifically, if $u$ has outdegree 1, say via arc $(u,v)$, then remove $u$;  remove also each arc $(w,u)$ and replace it by arc $(w,v)$.

Type 2:  Suppress a vertex with no directed path to a member of $X$.  More specifically, suppose $u$ is such a vertex.  Then, delete $u$, for each arc $(v,u)$ delete $(v,u)$, and for each arc $(u,v)$, delete $(u,v)$.  

The result of performing all possible simplifications of Type 1 or Type 2 on $N_p$ is denoted $T(N_p)$, called the \emph{standard form} of $N_p$.  

Figure 2 exhibits some rooted trees related to Figure 1.  In $M$ let the parent map $p$ satisfy $p(2) = B$ (and trivially $p(1) = B$, $p(3) = C$, $p(B) = A$, $p(C) = A$).  Then the tree $M_p$ is given in $a$.  Since $C$ has outdegree 1 in Fig 2a, it is suppressed by Type 1, resulting in the standard form $b$, which is $T(M_p)$.   Similarly if $p'$ is the parent map with $p'(2) = C$, then Fig 2c shows $T(M_{p'})$.  The tree $d$ is not displayed by $M$ since there is no parent map yielding $d$.  

It is easy to see that $N$ in Fig 1 also displays $b$ and $c$.  Consider the parent map $q$ for $N$ given by $q(2) = H$, $q(G) = E$, $q(H) = E$.  Then $N_q = e$ in Figure 2.   We simplify $e$ by suppressing $F$ by Type 2 and then $D$ and $G$ by Type 1.  Hence $T(N_q) = c$ in Figure 2.  For both the networks in figure 1, we have $Tr(M) = Tr(N) = \{b, c\}$ using the notation in Figure 2.

\section{ Reconstruction of regular networks from all their trees} %3

Suppose $\mathcal{D}$ is a nonempty collection of rooted trees each with the same base-set X.  In this section we present a procedure called MaximumProperChild (MPC) which constructs a phylogenetic network MPC($\mathcal{D}$) given $\mathcal{D}$.  The algorithm always terminates with a network.

The main theorem 3.1 asserts that if $\mathcal{D} = Tr(N)$ for some regular network $N$, then the output of the procedure is $N$, so $N$ has been reconstructed.  

\begin{thm} %Theorem 3.1.  
Suppose $N = (V,A,r,X)$ is a regular phylogenetic network.  Then the output of MaximumProperChild applied to $\mathcal{D} = Tr(N)$ is isomorphic with $N$; ie., $MPC(Tr(N)) \cong N$.
\end{thm}

An immediate consequence of Theorem 3.1 is Cor 3.2, which asserts that the set of trees displayed by a regular network uniquely determines the network.

\begin{cor} %Cor 3.2.  
Suppose $M$ and $N$ are regular phylogenetic networks with base-set $X$.  If $Tr(M) = Tr(N)$, then $M \cong N$.
\end{cor}

Theorem 3.1 need not be true without the assumption  that $\mathcal{D}$ includes all the trees displayed by $N$.  It is easy to find examples in which $N$ is not reconstructed if some trees are missing from $\mathcal{D}$.  On the other hand, it is also easy to find examples in which $\mathcal{D} \neq Tr(N)$ but still $MPC(\mathcal{D}) = N$.  What is important is that the ``right" trees lie in $\mathcal{D}$.  Roughly speaking, the ``right" trees are those that arise via the use of Lemma 3.3.  Further discussion of this point is in Section 5.  

The number of displayed trees may be exponentially large in $|X|$, so the algorithm need not be polynomial-time in $|X|$.  It is easy to see, however, that the procedure is polynomial-time in $|X| + |\mathcal{D}|$.  

Figure 1 shows that Cor 3.2 fails without the assumption of regularity.  

Here is an overview of the procedure:  We reconstruct the network $N$ recursively by finding the clusters of $N$.  Initially, we have only the root cluster, which is $X$.  At any given stage, given a cluster $C=cl(u,N)$ for some vertex $u\in V$ we are able to identify the clusters of all its children in $N$.  To do so, by construction we will know already the clusters along a directed path $X=P_n, P_{n-1}, \cdots, P_0 = C$  from $X$ to $C$.  The ``proper trees" for $C$, denoted $ProperTr(C)$,  will consist of any input trees $T$ exhibiting all the clusters along such a directed path already in our reconstruction.  We list the clusters for the children of $C$ in all the proper trees for $C$.  Among these we consider the set of ``maximal proper children,'' denoted $MaxProperCh(C)$, consisting of the clusters $U$ for children of $C$ in a proper tree, such that there is no other cluster $W$ which is a child of $C$ in some proper tree and for which $U \subset W \subset C$.  We show that these maximal children are necessarily the clusters of children of $C$ in $N$, and all clusters of the children of $C$ in $N$ arise in this manner.  Hence the children of $C$ are precisely the members of $MaxProperCh(C)$.  We insert the members of $MaxProperCh(C)$ into the set of vertices of our reconstruction, together with arcs from $C$ to each such vertex; then we continue recursively.  

An example of the procedure will be given in Section 4. 

The following is a precise more formal description:
\bigskip
\hrule height1pt
\medskip
\noindent\textbf{Algorithm} MaximumProperChild.  \\
\textbf{Input}:  $\mathcal{D}$ is a nonempty collection of rooted trees each with the base-set $X$.\\
\textbf{Output}:  a regular phylogenetic network $M$ with base-set $X$.\\
\textbf{Procedure}.We construct a sequence $M_0, M_1, \cdots$ of directed graphs where $M_k = (V_k, A_k)$.  Each member of $V_k$ is a nonempty subset of $X$, and $V_0 \subseteq V_1 \subseteq V_2 \subseteq \cdots$.

1.  Initially $M_0 = (V_0, A_0)$ with $V_0 = \{X\}$ and $A_0 = \emptyset$.  Thus $M_0$ has a single vertex which is the cluster $X$.   This vertex is not checked off.  

Recursively perform the following step 2:

2.  Suppose $M_k = (V_k, A_k)$ is known and some vertex $U \in V_k$ is not checked off.

2a.  If $V_k$ contains a singleton set $C = \{a\}$ which is not checked off, then $M_{k+1} = M_k$ except that $\{a\}$ has been checked off.

2b.  If $V_k$ contains a doubleton set $C = \{a,b\}$ which is not checked off, then $V_{k+1} := V_k \cup \{\{a\}, \{b\}\}$, and $A_{k+1} = A_k \cup \{(\{a,b\},\{a\}), (\{a,b\}, \{b\})\}$.  In $M_{k+1}$ check off all members of $V_k$ that were already checked off and in addition check off $\{a,b\}$, $\{a\}$, and $\{b\}$ but nothing else.   This thus adjoins the two singletons $\{a\}$ and $\{b\}$.  

2c.  Suppose neither 2a nor 2b applies.   Suppose $C \in V_k$ has not been checked off.  
Let $ProperTr(C) = \{T \in \mathcal{D}: C $ is a cluster of $T$ and there is a directed path $X=P_n, P_{n-1}, \cdots, P_1, P_0 = C$ in $T$  such that for each $i$, $cl(P_i,T)$ is a vertex of $M_k $ and each arc $(P_i, P_{i-1})$ is an arc of $M_k\}$ be the set of proper trees for $C$. 
Let $ProperCh(C) = \{D: $ for some $T \in ProperTr(C)$, $D$ is a child of $C\}$ be the set of children of $C$ in any proper tree for $C$.
Let $MaxProperCh(C) = \{D \in ProperCh(C): $ there is no $D' $ in $ProperCh(C)$ such that 
$D \subset D' \subset C \}$ (strict inclusions) be the set of maximal proper children of $C$.  
For each $D \in MaxProperCh(C)$, adjoin to $M$ the vertex $D$ (if it is not already present) and the arc $(C,D)$.  More explicitly define $V_{k+1} = V_k \cup \{D: D \in MaxProperCh(C)\}$.  Define $A_{k+1} = A_k \cup \{(D,C): D \in MaxProperCh(C)\}$.  In $M_{k+1}$ check off all vertices checked off in $M_k$ and also check off $C$ but nothing else.  Note that it is possible that $D$ is already present in $V_k$, but that this construction may still introduce a new arc incoming to $D$.  

3.  The procedure terminates with $M_n$ such that every member of $V_n$ has been checked off.  Return $M_n$.  
% When the procedure terminates with M_n, the output M will 
% be the Hasse diagram of V_n.  
\medskip
\hrule height1pt
\bigskip

It is clear that the procedure always terminates, whether or not $\mathcal{D} = Tr(N)$.  This is because $X$ is a finite set, so  $\mathcal{P}(X)$ is finite and there can only be finitely many vertices.  At the end of 2a, 2b, or 2c an additional vertex is checked off.  Hence after finitely many steps all vertices must be checked off.  

Moreover, whenever a new vertex $D$ is added in step 2b or 2c, $ProperTr(D)$ is nonempty.  This is trivially true if $D$ arose as a singleton set in 2b.  If D arose in 2c, then there exists a parent $C$ of $D$ and $T \in ProperTr(C)$.  Hence $T$ also lies in $ProperTr(D)$.  Thus when  2c is applied to $C$ containing at least three members of $X$, it identifies a child $D$ of $C$ which is a nonempty proper subset of $C$.  It follows that when the procedure terminates, each singleton set $\{x\}$ is in $V_n$.   

An example is given in the next section.  

We now turn to the proof of Theorem 3.1.  The first step is a lemma which identifies a useful tree related to  a given directed path in $N$.

\begin{lem}%Lemma 3.3.  
Let $C$ be a vertex of $N$.  Let $P_n = r, P_{n-1}, \cdots, P_1, P_0=C$ be a directed path in $N$ from the root $r$ to $C$.  There exists a tree $T$ displayed by $N$ in standard form such that, for $i = 0, \cdots, n$, $P_i$ is a vertex of $T$ and we have $cl(P_i,T) = cl(P_i,N)$.  
\end{lem}

\begin{proof}
We find a tree $T$ as follows:  The parent map $p$ which yields $T$ is selected by\\
(0) If $W$ is normal, $W\neq r$, then $p(W)$ is the unique parent of $W$. \\
(1) If $H$ is hybrid and $C < H$, choose a parent $p(H)$ of $H$ such that $C \leq p(H)$ in $N$. \\
%[This will ensure that in T, cl(C,T) = cl(C,N).]
(2) Suppose $n \geq 1$.  If $H$ is hybrid and $P_1 < H$, but it is false that $C < H$, choose $p(H)$ such that $P_1 \leq p(H)$ in $N$. \\
%[Together with (1), this will ensure that  cl(P1,T) = cl(P1,N).]
(3) Suppose $n \geq 2$.  If $H$ is hybrid and $P_2 < H$ but it is false that $P_1 < H$ (hence also false that $C < H$), then select $p(H)$ such that $P_2 \leq p(H)$ in $N$. \\
%[Together with (1) and (2), this ensures that cl(P2,T) = cl(P2,N).]
(k) In general, if $n \geq k$, $H$ is hybrid, and $P_k < H$ but it is false that $P_{k-1} < H$, select $p(H)$ such that $P_k \leq p(H)$ in $N$.\\
Since $P_n = r$, it follows that for each hybrid $H$, $p(H)$ will be defined.

I claim that $cl(C,N_p) = cl(C,N)$.  Clearly $cl(C,N_p) \subseteq cl(C,N)$.  Conversely, suppose $W$ is a vertex of $N$ and $C \leq W$ in $N$.  I will show that $C \leq W$ in $N_p$.  It suffices to show that whenever $C < W$ in $N$, then there exists a parent $P$ of $W$ in $N_p$ such that $C \leq P$ in $N$.  The result is immediate if $W$ has a unique parent $P$ in $N$ because since $C < W$ it follows  $C \leq P$.  If, instead, $W$ is hybrid, then by assumption $p(W)$ is a parent  of $W$  in $N_p$ and by (1) $C \leq p(W)$ in $N$.  This proves that $C \leq W$ in $N_p$ if $C \leq W$ in $N$.  Now, if $x \in cl(C,N)$ the choice $W=x$ shows, since $C\leq x$ in $N$, that $C \leq x$ in $N_p$, whence $x \in cl(C,N_p)$.  Thus $cl(C,N_p) = cl(C,N)$.  

Suppose $n \geq 1$.  I now claim that $cl(P_1,N_p) = cl(P_1,N)$.  It is immediate that $cl(P_1,N_p) \subseteq cl(P_1,N)$.  For the converse, suppose $x \in cl(P_1,N)$.   Suppose $W$ is a vertex of $N$ and $P_1 \leq  W$ in $N$.  I show that $P_1 \leq  W$ in $N_p$.  It suffices to show that if $P_1 < W$ in $N$, then there exists a parent $P$ of $P_1$ in $N_p$ such that $P_1 \leq  P$ in $N$.  If $C < W$ in $N$, then from above there exists a parent $P$ of $W$ in $N_p$ such that $C \leq  P$ in $N$, whence $P_1 \leq  C\leq P$ in $N$.  Hence we may assume that $C \nless W$ in $N$.   If $W$ is normal, then its unique parent $P$ must satisfy that $P \leq  W$ in $N_p$ (since arcs to normal vertices remain in $N_p$) whence $P_1 \leq  P$ in $N$.  If instead $W$ is hybrid, then since $P_1 < W$ but $C \nless W$ it follows from (2) that $p(W)$ satisfies $P_1 \leq  p(W)$ in $N$.  This proves that $P_1 \leq  W$ in $N_p$ if $P_1 \leq  W$ in $N$.  Now, if $x \in cl(P_1,N)$ the choice $W=x$ shows, since  $P_1\leq x$ in $N$, that $P_1 \leq  x$ in $N_p$, whence $x \in cl(P_1,N_p)$.  Thus $cl(P_1,N_p) = cl(P_1,N)$.  

The argument can be iterated to show that for $i = 0, \cdots, n$, $cl(P_i,N_p) = cl(P_i, N)$.  

Let $T = T(N_p)$  be the standard form of  $N_p$ obtained by suppressing vertices of outdegree 1 and vertices with no directed paths to any member of $X$.  By regularity of $N$, the sets $cl(P_i,N)$ are distinct for $i = 0, \cdots, n$.  Hence the sets $cl(P_i, N_p)$ are distinct for $i = 0, \cdots, n$.  I claim that $P_n, \cdots, P_0$ are vertices of $T$.  

Note first that there exists a directed path in $N$ of maximal length (number of arcs) starting at $P_0 = C$.  The path must end at some leaf which consists of a member $x \in X$ since $X$ contains all the leaves.  From (0) and (1) it follows that there is a path in $N_p$ from $C$ to $x$ as well; otherwise some vertex $W$ on that path would satisfy that  $C < W$ in $N$ so some parent $P$ of $W$ satisfies $C \leq  P$, but $p(W)$ satisfies that $C \nleq p(W)$, contradicting (0) or (1).  Hence there is a directed path in $N_p$ from $C$ to $x$, whence also a directed path from each $P_i$ to $x$.  It follows that no $P_i$ is suppressed because there is no path to a member of $X$.  

Moreover,  for $i = 1, \cdots, n$,  $P_i$ is a vertex of $T$; otherwise $P_i$ would have outdegree 1 in $N_p$ whence $cl(P_i,N_p) = cl(P_{i-1},N_p)$, whence $cl(P_i,N) = cl(P_{i-1},N)$.  Moreover, I claim that $C = P_0$ is a vertex of $T$.  The claim is immediate if $C$ is a leaf.  If $C$ is not a leaf then $C$ has  children  $D_1, D_2, \cdots, D_k$ in $N$, with $k\geq 2$.  By regularity $cl(D_j,N)$ is a proper subset of $cl(C,N)$.  If $C$ were not a vertex of $T$, it would have outdegree 1 in $N_p$.  Assume its child in $N_p$ is $D_1$.  Then $cl(C,N) = cl(C,N_p) = cl(D_1,N_p) \subseteq cl(D_1,N) \subset cl(C,N)$, a contradiction.  

It follows that in $T$ there is a directed path $P_n=X, P_{n-1}, \cdots, P_1, P_0 = C$  such that for $i =0, \cdots, n$,  $cl(P_i,T) = cl(P_i,N)$.  

\end{proof}

%Note that the proof of Lemma 3.3 showed that for each $i$, $P'_i = P_i$ 
%as a vertex in $N_p$ and also as a vertex in $T$.  

We now prove theorem 3.1.  

\begin{proof}%Proof.  
Let $N$ be a regular network and $\mathcal{D} = Tr(N)$.  Let the sequence of networks obtained from MPC be denoted $M_0, M_1, \cdots, M_n$ where $M_i = (V_i,A_i)$ has the set $V_i$ of vertices and the set $A_i$ of arcs.  Initially $V_0 = \{X\}$ and $A_0= \emptyset$.  

The proof will be by induction.  The $i$-th inductive hypothesis $H_i$ is that\\
(1) For each vertex $U$ of $M_i$ there exists a vertex $U'$ of $N$ such that $U = cl(U',N)$. \\
(2) For each arc $(U,W)$ of $M_i$, $(U',W')$ is an arc of $N$.\\
(3) For each vertex $U$ of $M_i$ that has at least one child in $M_i$, for every child $Y$ of $U'$ in $N$, there exists a vertex $W$ of $M_i$ such that $W$ is a child of $U$ in $M_i$ and $W' = Y$.

$H_0$ is trivially true since $X$ is the only vertex of $M_0$ and $X'$ is the root of $N$.

Claim 1.  Assume $H_j$ and the procedure has not terminated.   We show $H_{j+1}$.  

%Write $M$ for $M_j$.  
If 2a or 2b applies, then Claim 1 is immediate.  Hence we assume that 2c applies and there is a vertex $C$ of $M_j$ containing at least three points which has not been checked off.  Compute $ProperTr(C)$ and $MaxProperCh(C)$ as above.  By $H_j$, there exists vertex $C'$ of $N$ such that $C = cl(C',N)$.  It suffices to show that \\
(a) for each child $Y$ of $C'$ in $N$,  $D:=cl(Y,N)$ lies in $MaxProperCh(C)$; and \\
(b) each member of $MaxProperCh(C)$ consists of a cluster $D$ for which there exists a child $E$ of $C'$ in $N$ such that $D = cl(E,N)$.
%[OMIT (a) for each child Y of C' in N there exists D 
%in MaxProperCh(C) for which Y= D'.]

We first prove (a):

Claim 1a.  Let $Y$ be a child of $C'$ in $N$.  Then $D = cl(Y,N)$ is a member of $MaxProperCh(C)$.

Since $C$ is a vertex in $M_j$, there exists by $H_j$ a directed path  $r = P_n,$ $P_{n-1}$, $P_{n-2}$, $\cdots, P_1$, $P_0=C'$  in $N$ from $r$ to $C'$ such that for $i = 0, \cdots, n$,  $cl(P_i,N)$ is a vertex of $M$ and for $i = 1, \cdots, n$,  $(cl(P_i,N), cl(P_{i-1},N))$ is an arc of $M$.  (This is because $C$ occurred in $M$ as a child of some vertex, which occurred in $M$ as a child of some other vertex, etc.) 

By Lemma 3.3, since $Y$ is a child of $C'$ in $N$, there exists a tree $T$ in $Tr(N)$ that contains the directed path  $r = Q_n, Q_{n-1}, \cdots, Q_0, Q_{-1}$ for which $cl(Q_i,T) = cl(P_i,N)$, $cl(Q_0,T) = cl(C',N) = C$ and $cl(Q_{-1},T) = cl(Y,N) = D$.  By $H_j$, $T \in ProperTr(C)$, so it follows that  $D = cl(Y,N) \in ProperCh(C)$.  

I claim that $D \in MaxProperCh(C)$.  Otherwise, there exists a tree $\hat{T}$ in $ProperTr(C)$ with vertex $\hat{C}$  such that  $C = cl(\hat{C},\hat{T})$, $\hat{C}$  has child $\hat{D}$ in $T$,   and  $D \subset cl(\hat{D},\hat{T}) \subset C = cl(\hat{C},\hat{T})$.  Let $r = \hat{P}_m, \hat{P}_{m-1}, ... , \hat{P}_0 = \hat{C}$ be the directed path from the root $r$ to $\hat{C}$ in $\hat{T}$.   By construction, for $i = 0, \cdots, m$, $cl(\hat{P}_i,\hat{T})$ is a member of $M$ and for $i = 1, \cdots, m$, each arc $(cl(\hat{P}_i,\hat{T}), cl(\hat{P}_{i-1},\hat{T}))$ is an arc in $M_j$.  By $H_j$, for each $i$, $cl(\hat{P}_i,\hat{T})$ is a cluster of $N$; i.e.,  there exists vertex $Q_i$ in $N$ such that $cl(Q_i,N) = cl(\hat{P}_i,\hat{T})$ and $(Q_i,Q_{i-1})$ is an arc of $N$.  In particular, by regularity of $N$, $Q_0 = C'$.  

Let $\hat{p}$ be the parent map that yields $\hat{T}$ (i.e., $T(N_{\hat{p}}) = \hat{T}$).  Note for $0\leq i\leq m$ that $\hat{P}_i$ is a vertex of both $\hat{T}$ and $N_{\hat{p}}$.  Then $cl(Q_i,N) = cl(\hat{P}_i, \hat{T}) = cl(\hat{P}_i,N_{\hat{p}})  \subseteq cl(\hat{P}_i,N)$.  Since $N$ is regular it follows that $\hat{P}_i \leq  Q_i$ in $N$.  Since the arcs of $N_{\hat{p}}$ form a subset of the arcs of $N$, it follows from $\hat{P}_{i+1} \leq  \hat{P}_i$ in $N_{\hat{p}}$ that $\hat{P}_{i+1} \leq  \hat{P}_i$ in $N$ as well for $i = 0, \cdots, m-1$.  

Since $cl(\hat{P}_m, N_p)  = X = cl(Q_m, N)$ it is clear that $\hat{P}_m = Q_m$.  

Now 
$cl(Q_{m-1},N) = cl(\hat{P}_{m-1}, N_{\hat{p}}) \subseteq cl(\hat{P}_{m-1},N) \subseteq cl(\hat{P}_m,N) = cl(Q_m,N)$ [since $\hat{P}_m \leq  \hat{P}_{m-1}$ in $N$].  By regularity of $N$ it follows that  
$Q_m \leq  \hat{P}_{m-1} \leq  Q_{m-1}$ in $N$.  The arc $(Q_m, Q_{m-1})$ of $N$ is not redundant, so it follows that either $\hat{P}_{m-1} = Q_m$ or $\hat{P}_{m-1} = Q_{m-1}$.  But $\hat{P}_{m-1} \neq \hat{P}_m = Q_m$, so we see that $\hat{P}_{m-1} = Q_{m-1}$.

Similarly 
$cl(Q_{m-2},N) = cl(\hat{P}_{m-2}, N_{\hat{p}}) \subseteq cl(\hat{P}_{m-2},N) \subseteq cl(\hat{P}_{m-1},N) = cl(Q_{m-1},N)$ [since $\hat{P}_{m-1} \leq  \hat{P}_{m-2}$ in $N$].  By regularity of $N$ it follows that  
$Q_{m-1} \leq  \hat{P}_{m-2} \leq  Q_{m-2}$ in $N$.  The arc $(Q_{m-1}, Q_{m-2})$ of $N$ is not redundant since $N$ is regular, so it follows that either $\hat{P}_{m-2} = Q_{m-1}$ or $\hat{P}_{m-2} = Q_{m-2}$.  But $\hat{P}_{m-2} \neq \hat{P}_{m-1} = Q_{m-1}$, so we see that $\hat{P}_{m-2} = Q_{m-2}$.

In like manner we see that $\hat{P}_i = Q_i$ for $i = m-3, m-4, \cdots, 0$.  

It follows that  $\hat{C} = \hat{P}_0 = Q_0 = C'$.  Since $Y$ is a child of $C'$ in $N$ we know
$cl(Y,N) = D \subset cl(\hat{D},\hat{T}) = cl(\hat{D},N_{\hat{p}}) \subseteq cl(\hat{D},N) \subseteq cl(\hat{P}_0,N) = cl(C',N)$.  It follows that $C' \leq  \hat{D} \leq  Y$ in $N$.  Since the arc $(C',Y)$ is nonredundant, either $\hat{D} = C'$ or $\hat{D} = Y$.  But  $\hat{D} \neq C'$ since $\hat{D}$ is a child of $\hat{C} = C'$.  It follows that $\hat{D} = Y$.  Hence  $D = cl(Y,N) \subset cl(\hat{D},\hat{T}) \subseteq cl(\hat{D},N) = cl(Y,N)$, which is impossible.   This contradiction proves that $D =cl(Y,N) \in MaxProperCh(C)$.

Next we prove (b):

Claim 1b.  Each member $D$ of $MaxProperCh(C)$ satisfies that there exists a child $E$ of $C'$ in $N$ such that $D = cl(E,N)$. 

Let $D$ be a member of $MaxProperCh(C)$.   Thus there exists a tree $\hat{T}$ in $ProperTr(C)$ with vertex $\hat{C}$ such that  $C = cl(\hat{C},\hat{T})$ and  $\hat{C}$  has child $\hat{D}$  in $\hat{T}$ such that $cl(\hat{D},\hat{T}) = D$.  Let $r = \hat{P}_m, \hat{P}_{m-1}, \cdots, \hat{P}_0 = \hat{C}$ be the directed path from $r$ to $\hat{C}$ in $\hat{T}$.   By construction, for $i = 0, \cdots, m$, $cl(\hat{P}i,\hat{T})$ is a vertex of $M_j$ and for $i = 1, \cdots, m$, each arc $(cl(\hat{P}_i,\hat{T}), cl(\hat{P}_{i-1},\hat{T}))$ is an arc in $M_j$.  By $H_j$, for $i$ such that $0 \leq  i \leq  m$,  there exists a vertex $Q_i$ of $N$ such that $cl(\hat{P}_i,\hat{T}) = cl(Q_i,N)$, and for $1 \leq  i \leq  m$, $(Q_i, Q_{i-1})$ is an arc of $N$.  In particular, by regularity of $N$, $Q_0 = C'$.  

As in the proof of Claim 1a, we see that $\hat{P}_i = Q_i$ for $i = m, m-1, \cdots, 0$ and  $\hat{C} = \hat{P}_0 = Q_0 = C'$.  Since $\hat{D}$ is a child of $\hat{C}$ in $N_{\hat{p}}$, it follows that $\hat{C} \leq  \hat{D}$ in $N_{\hat{p}}$, whence $\hat{C} \leq  \hat{D}$ in $N$.  Since $\hat{D} \neq \hat{C}$, there exists a child $E$ of $\hat{C}=C'$ in $N$ such that  $E \leq  \hat{D}$.    Hence $D = cl(\hat{C},\hat{T}) = cl(\hat{C}, N_{\hat{p}}) \subseteq cl(E,N)$.   By Claim 1a, $cl(E,N) \in MaxProperCh(C)$.  Hence $D$ is not in $MaxProperCh(C)$ unless $D = cl(E,N)$, proving Claim 1b.

This completes the proof of Claim 1. 

We now complete the proof of Theorem 3.1.  

We saw above that the procedure terminates, say with $M_n$.  By Claim 1,  $H_n$ will be true.  In fact, each vertex $W$ of $N$ has been represented in $M_n$ in the sense that $cl(W,N) \in V_n$.  To see this, note that there is a directed path $P_0=r, P_1, \cdots, P_k = W$ in $N$ since $r$ is the root of $N$.  Since $X' = r$, by Claim 1a it follows that $cl(P_1,N)$ is a member of $MaxProperCh(X)$, whence by construction $cl(P_1,N) \in V_n$, $cl(P_1,N)' = P_1$, and $(cl(r,N), cl(P_1,N))$ in $A_n$.  Since $cl(P_1,N)'=P_1$ and $P_2$ is a child of $P_1$ in $N$, by Claim 1a again it follows that $cl(P_2,N) \in V_n$, $cl(P_2,N)' = P_2$, and $(cl(P_1,N), cl(P_2,N)) \in A_n$.  Repeating the argument we ultimately obtain that $cl(P_k,N) = cl(W,N)$ in $V_n$.  

Since every arc in $N$ lies on some directed path in $N$ starting at $r$ hence occurs as some arc $(P_i,P_{i+1})$ using the notation above, the same argument shows that the arc corresponds to the arc $(cl(P_i,N), cl(P_{i+1},N)) \in A_n$. Thus every vertex and arc of $N$ has a corresponding vertex and arc in $M_n$.

There remains only to show that $M_n$ has no additional vertices or arcs.  By Claim 1b every vertex which is added at any stage has the form $cl(E,N)$ for some vertex $E$ of $N$.  Hence $M_n$ has no additional vertices.  By claim 1a, every arc in $M_n$  corresponds to an arc in $N$.

This completes the proof.  

\end{proof}

\section{ An example of the reconstruction }%4

Let $N$ be the network given in Figure 3.  The base-set is $X = \{1,2,3,4,5,6\}$.  The clusters satisfy $cl(A) = X$, $cl(B) = \{1,2,3,4,6\}$, $cl(C) = \{5,6\}$, $cl(D) = \{1,2,3,6\}$, $cl(E) = \{2,3\}$, $cl(F) = \{1,2,6\}$, $cl(G) = \{1,2,3\}$, and $cl(i) = \{i\}$ for $1\leq i\leq 6$.  An inspection shows that $N$ is regular.

\begin{figure}[!htb]  %3
\begin{center}

\begin{picture}(180,170) 

\put(100,160) {\vector(1,-1){30}}
\put(100,160) {\vector(-1,-1){60}}
\put(40,100) {\vector(-1,-1){30}}
\put(40,100) {\vector(0,-1){60}}
\put(130,130) {\vector(-1,-1){30}}
\put(130,130) {\vector(1,-1){30}}
\put(100,100) {\vector(-1,-1){30}}
\put(100,100) {\vector(1,-1){30}}
\put(70,70) {\vector(-1,-1){30}}
\put(70,70) {\vector(0,-1){60}}
\put(70,70) {\vector(1,-1){30}}
\put(130,70) {\vector(-1,-1){30}}
\put(130,70) {\vector(1,-1){30}}
\put(160,40) {\vector(0,-1){30}}
\put(160,40) {\vector(-3,-1){90}}
\put(110,160){$A$}
\put(140,130){$B$}
\put(30,100){$C$}
\put(110,100){$D$}
\put(170,40){$E$}
\put(80,70){$F$}
\put(140,70){$G$}
\put(100,30){1}
\put(70,0){2}
\put(160,0){3}
\put(170,100){4}
\put(10,60){5}
\put(40,30){6}

\end{picture}

\caption{  A regular network $N$ with $X = \{1,2,3,4,5,6\}$ which will be reconstructed from its trees.}
\end{center}
\end{figure}
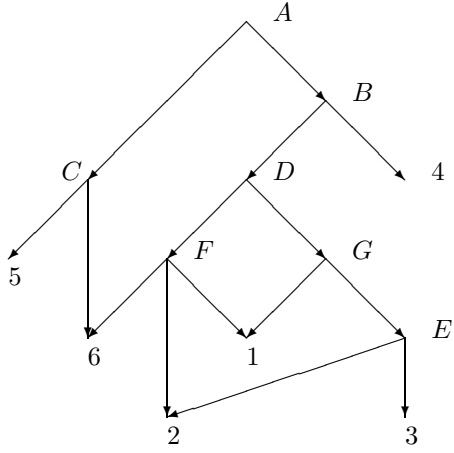

There are three hybrid vertices 1, 2, 6, each with indegree 2.  Hence there are 8 parent maps.  Here I will list the displayed trees by telling the parent map and the nontrivial clusters of each:  \\
$T_1$: $p(1) = G$, $p(2) = E$, $p(6) = F$. Clusters $\{2,3\}$, $\{1,2,3\}$, $\{1,2,3,6\}$, $\{1,2,3,4,6\}$.  \\
$T_2$: $p(1) = G$, $p(2) = E$, $p(6) = C$.  Clusters $\{2,3\}$, $\{1,2,3\}$, $\{1,2,3,4\}$, $\{5,6\}$.\\
$T_3$: $p(1) = G$, $p(2) = F$, $p(6) = F$.  Clusters $\{1,3\}$, $\{2,6\}$, $\{1,2,3,6\}$, $\{1,2,3,4,6\}$.\\
$T_4$: $p(1) = G$, $p(2) = F$, $p(6) = C$.  Clusters $\{1,3\}$, $\{1,2,3\}$, $\{1,2,3,4\}$, $\{5,6\}$.\\
$T_5$: $p(1) = F$, $p(2) = E$, $p(6) = F$.  Clusters $\{2,3\}$, $\{1,6\}$, $\{1,2,3,6\}$, $\{1,2,3,4,6\}$.\\
$T_6$: $p(1) = F$, $p(2) = E$, $p(6) = C$.  Clusters $\{2,3\}$, $\{1,2,3\}$, $\{1,2,3,4\}$, $\{5,6\}$.\\
$T_7$: $p(1) = F$, $p(2) = F$, $p(6) = F$.  Clusters $\{1,2,6\}$, $\{1,2,3,6\}$, $\{1,2,3,4,6\}$.\\
$T_8$: $p(1) = F$, $p(2) = F$, $p(6) = C$.  Clusters $\{1,2\}$, $\{1,2,3\}$, $\{1,2,3,4\}$, $\{5,6\}$.  

We now perform procedure MaximumProperChild.  Let $M_k = (V_k, A_k)$.  Initially $V_0 = \{X\}$.  
The proper children of $X$ are the children of $X$ in any proper tree.  All the trees are proper trees for $X$.  Hence 
$ProperCh(X) = \{\{1,2, 3,4,6\}, \{5\}, \{1,2,3,4\}, \{5,6\}\}$.  
The maximal proper children are the maximal members of $ProperCh(X)$.  Hence 
$MaxProperCh(X) = \{\{1,2,3,4,6\},$  $\{5,6\}\}$.  These are adjoined to $M_0$ as children of $X$.
Hence $M_1 = (V_1, A_1)$ has $V_1 = \{X, \{1,2,3,4,6\}, \{5,6\}\}$ and has arcs $(X , \{1,2,3,4,6\})$ and $(X, \{5,6\})$.   

Let $C =\{5,6\}$ in $V_1$.  By 2b, the children will be $\{5\}$ and $\{6\}$.  Hence $M_2$ has
$V_2 = \{X, \{1,2,3,4,6\}, \{5,6\}, \{5\}, \{6\}\}$ and the arcs are those of $M_1$ together with $(\{5,6\},\{5\})$ and $(\{5,6\}, \{6\})$.

Let $C = \{1,2,3,4,6\}$.  The proper trees must contain both $X$  and $\{1,2,3,4,6\}$.  Hence
$ProperTr(C) = \{T_1, T_3, T_5, T_7\}$.  The proper children of $C$ are the children of $C$ in one of the proper trees.  Hence 
$ProperCh(C) = \{\{1,2,3,6\}, \{4\}\}$.  In this case all proper children are maximal.  Hence $M_3$ has $V_3 = V_2$ $\cup \{\{1,2,3,6\}$, $\{4\}\}$ and suitable arcs are also added.

Let $C = \{1,2,3,6\}$.  A proper tree must contain $C$,  some parent  of $C$ hence $\{1,2,3,4,6\}$, and $X$.   Thus $ProperTr(C) = \{T_1,T_3,T_5,T_7\}$.  The proper children are the children of $C$ in any of these proper trees, so  $ProperCh(C) = \{\{1,2,3\}$, $\{6\}, \{1,3\}$, $\{2,6\}, \{1,6\}$, $\{2,3\}$, $\{1,2,6\}, \{3\}\}$.  Then $MaxProperCh(C)$ $ = \{\{1,2,3\}$, $\{1,2,6\}\}$.  These are adjoined, so $V_4 = V_3 \cup \{\{1,2,3\}, \{1,2,6\}\}$ and arcs are inserted so that these are the children in $M_4$ of $\{1,2,3,6\}$.  

Let $C = \{1,2,3\}$.  A proper tree must contain $\{1,2,3\}$, $\{1,2,3,6\}$, $\{1,2,3,4,6\}$, and $X$.  Hence $ProperTr(C) = \{T_1, T_6\}$.  Then $ProperCh(C) = \{\{1\}, \{2,3\}\} =  MaxProperCh(C)$.  Now $V_5 = V_4 \cup \{\{1\}, \{2,3\}\}$.  

Let $C = \{1,2,6\}$.  A proper tree must contain $\{1,2,6\}, \{1,2,3,6\}, \{1,2,3,4,6\}$, and $X$.  Hence $ProperTr(C) = \{T_7\}$  It follows that  
$ProperCh(C) = \{\{1\}$, $\{2\}, \{6\}\}$ $ = MaxProperCh(C)$.  Now $V_6 = V_5 \cup \{\{1\}, \{2\}, \{6\}\}$.   Note that $\{6\}$ was already in $V_5$, but it is at this stage that we obtain the arc $(\{1,2,6\}, \{6\})$.  

Let $C = \{2,3\}$.  By 2b the children will be $\{2\}$ and $\{3\}$.  Hence $V_7 = V_6 \cup \{\{2\}, \{3\}\}$.  

The procedure terminates now with $M_7$.  Note that $V_7$ now consists of exactly the sets $cl(U,N)$ where $U$ is a vertex of $N$.  Similarly the arcs of $M_7$ consist exactly of the arcs $(cl(U,N), cl(W,N))$ such that $(U,W)$ is an arc of $N$.  Thus $M_7$ is isomorphic with $N$; indeed, it is the cover digraph of $N$.  

It is natural to wish that the identification of the children could be simplified, for example by merely looking at the maximal children of $C$ in any input tree $T$ rather than insisting on proper trees for $C$.  This alternative approach, however, fails on this example.    If we did not insist on proper trees, then  $\{1\}$ is not a maximal child of $\{1,2,3\}$ since $T_8$ contains $\{1,2,3\}$ with the child $\{1,2\}$.  Our procedure works since $T_8$ is not a proper tree for $\{1,2,3\}$ because the parent of $\{1,2,3\}$ in $T_8$ is $\{1,2,3,4\}$ which had not been identified as a cluster in $N$.

\section{Discussion }%5

The main result in this paper is that, if $N= (V,A,r,X)$ is a regular network, then the procedure MaximumProperChild will reconstruct $N$ from the collection $Tr(N)$ of all trees displayed by $N$.  The definition of regularity has two parts: \\
(1) $cl_N: V \to \mathcal{P}(X)$  is one-to-one; and \\
(2) there is an arc $(u,v) \in A$ iff $cl_N(v) \subset cl_N(u)$ and there is no $w \in V$ such that
$cl_N(v) \subset cl_N(w) \subset cl_N(u)$.

The network $N$ in Figure 1 is not regular because (1) fails, and the method fails to reconstruct $N$.  Both $M$ and $N$ in Figure 1 have the same displayed trees.  The procedure  of course reconstructs $M$ since $M$ is regular.

Figure 4 shows two networks $A$ and $B$ that are not regular.  Even though (1) holds for both, (2) fails.  It can be seen that they display exactly the same trees. Indeed, using Newick notation, they both display $(4,(3,(1,2)))$, $(4, (1,(2,3)))$, and $(4, (1,2,3))$ but not $(4,(2,(1,3)))$.  Hence the conclusion of Cor 4.2 fails for these networks when just the second condition of regularity fails.  Curiously, one easily checks that there is no regular network $N$ that displays these three trees and no others.  
  
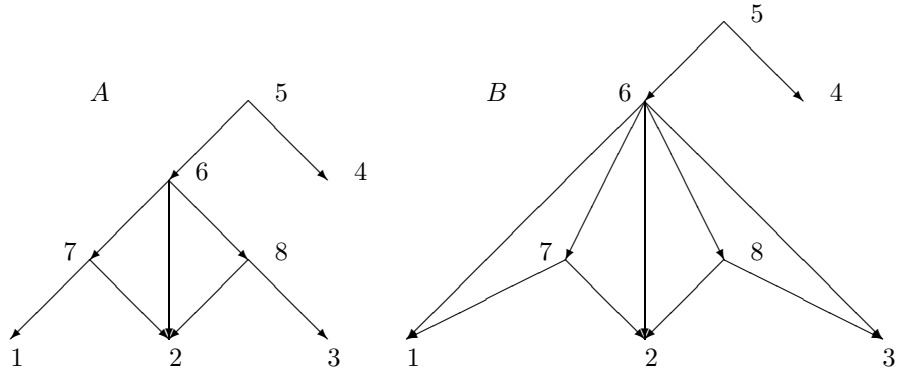
\begin{figure}[!htb]  %4 revised
\begin{center}

\begin{picture}(350,140) 

\put(100,100) {\vector(1,-1){30}}
\put(100,100) {\vector(-1,-1){30}}
\put(70,70) {\vector(-1,-1){30}}
\put(70,70) {\vector(0,-1){60}}
\put(70,70) {\vector(1,-1){30}}
\put(40,40) {\vector(-1,-1){30}}
\put(40,40) {\vector(1,-1){30}}
\put(100,40) {\vector(-1,-1){30}}
\put(100,40) {\vector(1,-1){30}}
\put(10,0) {1}
\put(70,0){2}
\put(130,0){3}
\put(140,70){4}
\put(110,100){5}
\put(80,70){6}
\put(30,40){7}
\put(110,40){8}
\put(40,100){$A$}

\put(280,130) {\vector(-1,-1){30}}
\put(280,130) {\vector(1,-1){30}}
\put(250,100) {\vector(-1,-1){90}}
\put(250,100) {\vector(-1,-2){30}}
\put(250,100) {\vector(0,-1){90}}
\put(250,100) {\vector(1,-2){30}}
\put(250,100) {\vector(1,-1){90}}
\put(220,40) {\vector(-2,-1){60}}
\put(220,40) {\vector(1,-1){30}}
\put(280,40) {\vector(-1,-1){30}}
\put(280,40) {\vector(2,-1){60}}
\put(160,0){1}
\put(250,0){2}
\put(340,0){3}
\put(320,100){4}
\put(290,130){5}
\put(240,100){6}
\put(210,40){7}
\put(290,40){8}
\put(190,100){$B$}

\end{picture}

\caption{ $A$ and $B$ are non-regular networks with $X = \{1,2,3,4\}$ that display the same trees.}
\end{center}
\end{figure}

In the example of section 4, it is easy to see that $N$ would be reconstructed using procedure MaximumProperChild given the input $\mathcal{D} = \{T_1, T_2, T_7\}$.  Thus, reconstruction may be possible even if $\mathcal{D} \ne Tr(N)$.  On the other hand, in the same example if $\mathcal{D} = \{T_1, T_2, T_3\}$, then $N$ is not reconstructed. It would be interesting to characterize $\mathcal{D} \subset Tr(N)$ for which reconstruction of $N$ is possible using MaximumProperChild with input $\mathcal{D}$.

\end{document}